\documentclass[11pt,a4paper]{amsart}
\usepackage[numbers]{natbib}
\usepackage[top=2.5cm,bottom=3cm,left=3cm,right=3cm,marginparwidth=1.75cm]{geometry}

\usepackage[english]{babel}

\newtheorem{theorem}{Theorem}
\newtheorem{problem}{Problem}

\newtheorem{lemma}[theorem]{Lemma}

\newtheorem{corollary}[theorem]{Corollary}
\theoremstyle{definition}

\newtheorem{definition}[theorem]{Definition}
\theoremstyle{remark}
\newtheorem{remark}[theorem]{Remark}
\newtheorem{example}[theorem]{Example}

\usepackage{subcaption}
\usepackage{tabularx}
\usepackage{amsmath,amssymb,bm,color}
\usepackage[english]{babel}
\usepackage{mathtools}
\usepackage{enumerate}

\usepackage{subcaption}
\usepackage{comment}
 
\usepackage{multirow}
\usepackage{xcolor}
\usepackage{algorithm}
\usepackage[noend]{algpseudocode}

\usepackage{pgfplots,relsize}
\usepackage{mleftright}

\usepackage{todonotes}

\usepackage{algorithm}
\usepackage[noend]{algpseudocode}
\algdef{SE}[DOWHILE]{Do}{doWhile}{\algorithmicdo}[1]{\algorithmicwhile\ #1}%

\usepackage[acronym,shortcuts]{glossaries}

\definecolor{my_green}{RGB}{77,175,74}
\definecolor{my_red}{RGB}{228,26,28}
\definecolor{my_blue}{RGB}{55,126,184}
\definecolor{my_orange}{RGB}{255,127,0}
\definecolor{my_lila}{RGB}{152,78,163}

\usepackage{pgfplots}
\pgfplotsset{compat=newest,
	complexityPlotStyle/.style={
		legend columns=1,
		width = 0.75\textwidth,
		label style={font=\footnotesize},
		xlabel={code rate $R$},
		xmax = 1.0,
		xmin = 0.0,
		tick label style = {font=\footnotesize},
		ylabel={value of complexity coefficient},
		ymax = 1,
		xmajorgrids,
		ymajorgrids,
		grid style=dashed,
		legend style={at={(1,1)},anchor=north east,font=\scriptsize,legend cell align=left, align=left, draw=black},
		line width=1pt,
	},
}

\DeclareMathOperator*{\argmax}{argmax}
\usepackage[title]{appendix} 
\makeatletter
\newcommand{\@chapapp}{\relax}%
\makeatother

\DeclareMathOperator{\Cov}{Cov}

\newcommand{\F}{\mathbb{F}}

\newcommand{\Prob}{P} 

\newcommand{\wt}{\text{wt}}

\title{Generic Decoding in the Cover Metric}

	\author[S. Bitzer]{Sebastian Bitzer}
	\address{Department of Electrical and Computer Engineering\\
		Technical University of Munich\\
 Germany
	}
	\email{sebastian.bitzer@tum.de}

	\author[J. Renner]{Julian Renner}
	\address{Department of Electrical and Computer Engineering\\
		Technical University of Munich\\
 Germany
	}
	\email{julian.renner@tum.de}

	\author[A. Wachter-Zeh]{Antonia Wachter-Zeh}
	\address{Department of Electrical and Computer Engineering\\
		Technical University of Munich\\
 Germany
	}
	\email{antonia.wachter-zeh@tum.de}

	\author[V. Weger]{Violetta Weger}
	\address{Department of Electrical and Computer Engineering\\
		Technical University of Munich\\
 Germany
	}
	\email{violetta.weger@tum.de}

\begin{document}

\maketitle

\begin{abstract}
In this paper, we study the hardness of decoding a random code endowed with the cover metric. As the cover metric lies in between the    Hamming and rank metric, it presents itself as a promising candidate for code-based cryptography.
We give a polynomial-time reduction from the classical Hamming-metric decoding problem, which proves the NP-hardness of the decoding problem in the cover metric.
We then provide a generic decoder, following the information set decoding idea from Prange's algorithm in the Hamming metric. A study of its cost then shows that the complexity is exponential in the number of rows and columns, which is in contrast to the behaviour in the Hamming metric, where the complexity  grows exponentially in the number of code symbols.
\end{abstract}

\section{Introduction}\label{sec:intro}
The recent advances in quantum computers have accelerated the process of developing efficient and quantum-resilient public-key cryptosystems. In the current NIST standardization process for post-quantum public-key cryptography, lattice-based and code-based systems are dominating. While the only code-based finalist of the competition, \emph{Classic McEliece}, is well-trusted in the community due to its resilience against attacks for more than forty years, at the same time its key sizes are huge compared to nowadays employed public-key cryptography. It is therefore a significant task to develop code-based public-key systems of smaller key sizes. 

One approach towards this goal is to change the underlying metric of the coding scheme from the Hamming metric to other metrics. For example, the rank metric \cite{delsarte, gabidulin, roth2} has shown large potential in the NIST Round 2 candidates RQC \cite{RQC} and ROLLO \cite{ROLLO}. Recent investigations also indicate that the Lee metric is a promising candidate for designing code-based systems \cite{Z4,thomas, leenp}. 

In this paper, we investigate the generic decoding of another metric: the cover metric. Similar to the rank metric, the words that we consider are now matrices instead of vectors as for the Hamming metric. The cover weight of a matrix is the minimum number of lines (rows and columns) that are needed to cover all non-zero elements of the matrix. 
The cover metric has so far mainly been used for crisscross error correction \cite{gabidulin1985optimal, gabidulin1972lattice, roth2, roth,Wachterzeh_Crisscross_journal}. 

A code-based system can be built based on the McEliece principle in any metric, i.e., the structure of a code is hidden such that encryption is done by multiplying with a generator matrix that looks random whereas efficient decryption can be done by the legitimate receiver who knows how to retrieve the structure of the code. For an attacker, the task of recovering the message from the ciphertext is equivalent to decoding in a random code. Thus, in order to build a code-based cryptosystem with the McEliece principle, it is important to investigate \emph{generic decoding}, i.e., decoding in a certain metric without knowing the code that the message was encoded with. 

In this work, we first recall properties of codes in the cover metric, which are $\mathbb{F}_q$-linear codes, such as the size of a cover ball, an asymptotic Gilbert--Varshamov bound, and derive new results such as the density of maximum cover distance codes. We prove that the problem of finding a low cover weight  codeword and the cover decoding problem are NP-complete. 
While for $\mathbb{F}_q$-linear codes endowed with the rank metric, this problem, known as the MinRank problem, is known to be NP-complete as well \cite{minrank},  note that for $\mathbb{F}_{q^m}$-linear codes endowed with the rank metric such a result could not yet be proven~\cite{ranknp}.

As an information set decoding equivalent in the new metric, we provide a Prange-like generic decoding algorithm in the cover metric.
We derive its complexity and compare it with generic decoding in the Hamming metric.
This comparison shows a significant difference:
in the Hamming metric, the complexity grows exponentially in the number of code symbols, whereas in the cover metric the complexity is exponential in the number of lines.

This paper is organized as follows. In Section \ref{sec:preliminaries}, we introduce the notation and recall the  definitions of cover-metric codes that are essential for the remainder of the paper. In Section \ref{sec:bounds}, we provide new density results of maximum cover-distance codes. We show the NP-completeness of the cover-metric decoding problem in Section \ref{sec:NP} and provide the Prange-like generic decoder in Section \ref{sec:generic}. Finally, we conclude this paper in Section \ref{sec:conclusion}.

\section{Preliminaries}\label{sec:preliminaries}
Let us first introduce the notation that we use throughout this paper.
Let $q$ be a prime power and let $\mathbb{F}_q$ denote the finite field with $q$ elements. Let $m$ and $n$ be positive integers and denote by $\mathbb{F}_q^{m \times n}$ the space of $m \times n$ matrices over $\mathbb{F}_q$. By the term lines we denote rows and columns of a matrix. We denote by $\text{GL}_m(\mathbb{F}_q)$ the general linear group, that is all invertible $m \times m$ matrices over $\mathbb{F}_q$. For a matrix $A \in \mathbb{F}_q^{k \times n}$ we denote by $\text{rk}(A)$ its rank and by $\langle A \rangle$ its rowspan.
The identity matrix of size $m$ will be denoted by $\text{Id}_m$. For a set $S$, we will denote its cardinality by $\lvert S \rvert$. 
For a matrix $A \in \mathbb{F}_q^{m \times n}$ and a block $\mathcal{I} \times \mathcal{J} \subseteq \{1, \ldots, m\} \times \{1, \ldots, n\}$ of size $\tilde{m} \times \tilde{n}$ we denote by $A_{\mathcal{I} \times \mathcal{J}} \in \mathbb{F}_q^{\tilde{m} \times \tilde{n}}$ the matrix consisting of the entries of $A$, where the rows are indexed by $\mathcal{I}$ and the columns are indexed by $\mathcal{J}.$ Similarly, for a matrix $A \in \mathbb{F}_q^{m \times n}$, respectively for a vector $a \in \mathbb{F}_q^n$ and a set $I \subseteq \{1, \ldots, n\}$ of size $k$ we denote by $A_I \in \mathbb{F}_q^{m \times k}$ the matrix consisting of the columns of $A$ indexed by $I$, respectively by $a_I \in \mathbb{F}_q^k$ the vector consisting of the entries of $a$ indexed by $I.$ Thus, calligraphic letters are used to  indicate a block, while normal letters indicate a subset of $\{1, \ldots, n\}$.
 \begin{definition}[Matrix Code]
 A linear $[m \times n, k]$ cover metric code $\mathcal{C}$  over $\mathbb{F}_q$ is a $k$-dimensional linear subspace of $\mathbb{F}_q^{m \times n}$.
 \end{definition}
Let $\mathcal{C} \subseteq \mathbb{F}_q^{m \times n}$ be linear code of dimension $k$. 
Then, we say that $\mathcal{C}$ has generators $G_1, \ldots, G_k  \in \mathbb{F}_q^{m \times n}$ if 
 $$\mathcal{C} = \left\{ \sum\limits_{i=1}^k u_i G_i \mid u_i \in \mathbb{F}_q \right\}.$$
 Analogously, we may write any codeword $C=  \sum\limits_{i=1}^k u_i G_i \in \mathbb{F}_q^{m \times n}$ as
 $$ C= \begin{pmatrix} u_1 \text{Id}_m &  u_2\text{Id}_m &   \cdots & & u_k \text{Id}_m \end{pmatrix} \begin{pmatrix} G_1 \\ G_2 \\ \vdots \\ G_k \end{pmatrix}.$$
  We define the dual code with respect to the trace product, that is for matrices $C,D \in \mathbb{F}_q^{m \times n}$ we have
 $$\langle C,D \rangle = \text{Tr}(CD^\top),$$ where $\text{Tr}(A)$ denotes the trace of  a matrix $A.$
 Then, the dual code of $\mathcal{C} \subseteq \mathbb{F}_q^{m \times n}$ of dimension $k$ is defined as 
 $$\mathcal{C}^\perp = \{ D \in \mathbb{F}_q^{m \times n} \mid \langle C,D \rangle =0 \ \forall C \in \mathcal{C}\}.$$
 The dual code is of dimension $mn-k$. 
 We say that $\mathcal{C}^\perp$ has generators $H_1, \ldots,$ $ H_{mn-k} \in \mathbb{F}_q^{m \times n}$ if 
 $$\mathcal{C} = \left\{ C \in \mathbb{F}_q^{m \times n} \mid  \langle C, H_i \rangle = 0 \  \forall i \in \{1, \ldots, mn-k\} \right\}.$$
 A syndrome of $E \in \mathbb{F}_q^{m \times n}$ through $H_1, \ldots, H_{mn-k}$ is then  defined as $$( \langle H_1, E \rangle, \ldots, \langle H_{mn-k},E \rangle ) \in \mathbb{F}_q^{mn-k}.$$
 \begin{definition}[Cover]
Let $A =\left( a_{i,j} \right)_{\substack{1 \leq i \leq  m , \\ 1 \leq j \leq n}} \in \mathbb{F}_q^{m \times n}$. We say that $\Cov(A) \subset \{1, \ldots, m+n\}$ is a cover of $A$ if for all $i \in \{1, \ldots, m\}$ and for all $j \in \{1, \ldots, n\}$ we have that, if $a_{i,j} \neq 0$, then either $i \in \Cov(A)$ or $j+m \in \Cov(A)$.
 \end{definition}
 Observe that a cover of a matrix does not have to be unique. In fact, we call a cover minimal if it is a cover of minimal size, i.e., for $A  \in \mathbb{F}_q^{m \times n}$, the cover $C \subset \{1, \ldots, m+n\}$ is a minimal cover of $A$, if 
 $$ \lvert C \rvert = \min\{ \lvert \Cov(A) \rvert \quad  \mid \Cov(A) \quad \text{is a cover of} \quad A\}. $$
The cover weight of a matrix is then the size of a minimal cover of this matrix.
 \begin{definition}[Cover Weight]
 Let $A  \in \mathbb{F}_q^{m \times n}$. The cover weight of $A$ is defined as
 $$\textup{wt}_C(A) = \min\{ \lvert \Cov(A) \rvert \quad  \mid \Cov(A) \quad \text{is a cover of} \quad A\}. $$
 \end{definition}
 With this we can define the cover distance of two matrices, which defines a metric on $\mathbb{F}_q^{m \times n}$.
  \begin{definition}[Cover Distance]
 Let $A,B  \in \mathbb{F}_q^{m \times n}$. The cover distance between $A$ and $B$ is defined as
 $$d_C(A,B) = \textup{wt}_C(A-B). $$
 \end{definition}
We can also define the minimum cover distance of a code.
 \begin{definition}[Minimum Cover Distance]
Let $\mathcal{C}$  be a linear $[m \times n, k]$ cover metric code  over $\mathbb{F}_q$. The minimum cover distance of $\mathcal{C}$ is defined as
$$d_C= \min \{\textup{wt}_C(A) \mid A \in \mathcal{C}, A \neq 0\}.$$
 \end{definition}
We sometimes refer to an  $[m \times n, k]$ cover-metric code with minimum cover distance $d$ as an $[m \times n, k,d]$ cover-metric code.

For $x \in \mathbb{F}_{q^m}^n$, and $\Gamma=\{\gamma_1, \ldots, \gamma_m\}$ a basis of $\mathbb{F}_{q^m}$ over $\mathbb{F}_q$, define  $\Gamma(x) \in \mathbb{F}_q^{m \times n}$, to be the associated matrix to $x$, i.e. such that
$$x_i= \sum_{j=1}^m \Gamma_{i,j}(x) \gamma_j,$$ for all $i \in \{1, \ldots, n\}.$
Then, one clearly has that 
$$\wt_R(\Gamma(x)) \leq \wt_C(\Gamma(x)) \leq \wt_H(x).$$
We can also define the cover metric directly over $\mathbb{F}_{q^m}^n$ and consider $\mathbb{F}_{q^m}$-linear codes.
The weight of a vector $x \in \mathbb{F}_{q^m}^n$ is then the cover weight of the associated matrix $\Gamma(x) \in \mathbb{F}_q^{m \times n}$, i.e.,
$$\wt_C(x)= \wt_C(\Gamma(x)).$$
Clearly, we still have the connection to the other two metrics, that is for $x \in \mathbb{F}_{q^m}^n$ we have
$$\wt_R(x) \leq \wt_C(x) \leq \wt_H(x).$$
However, the cover weight of $x \in \mathbb{F}_{q^m}$ is highly dependent on the chosen basis $\Gamma$.
 \begin{example}
 Let us consider $\mathbb{F}_8=\mathbb{F}_2[\alpha]$ for some primitive element $\alpha$ and the vector $x = (\alpha+1,\alpha+1,0)$. For the basis $\Gamma=\{1,\alpha,\alpha^2\}$ we have that  the associated matrix
 $\Gamma(x)=\begin{pmatrix} 1 & 1 & 0 \\ 1 & 1 & 0 \\ 0 & 0 & 0 \end{pmatrix}$ has cover weight 2. However, for the basis $\Gamma'=\{1+\alpha, \alpha, \alpha^2\}$ we get the associated matrix
 $\Gamma'(x)= \begin{pmatrix} 1 & 1 & 0 \\ 0 & 0 & 0 \\ 0 & 0 & 0 \end{pmatrix}$ of cover weight 1.
 \end{example}
In fact, even more can be shown. 
\begin{lemma}
 Let $x \in \mathbb{F}_{q^m}^n$. Then there exists a basis $\Gamma$, such that $$\textup{wt}_C(\Gamma(x))=\textup{rk}(\Gamma(x))=\textup{wt}_R(x).$$
\end{lemma}
\begin{proof}
Finding the basis $\Gamma$ such that the associated matrix $\Gamma(x)$ has the lowest cover weight can be achieved by performing Gaussian elimination to get a row-reduced $\Gamma(x)'$. This operation is equivalent to  multiplying the associated matrix with an invertible matrix $U \in \text{GL}_{m}(\mathbb{F}_q)$.  Finally, since such operation represents a change of basis, there exists a basis $\Gamma'$,  such that the associated matrix has cover weight equal to the rank of $x$ in the extension field. 
\end{proof}
Since  the cover weight over the extension field would depend on the chosen basis and there always exists a basis which would recover the rank weight, we will rather remain in the base field $\mathbb{F}_q$ and consider matrix codes than pursuing the $\mathbb{F}_{q^m}$-linear approach.

 To describe the operations which are performed on the generators and the received matrix for generic decoding of a matrix code, it is more convenient when the associated vector code is considered.
 To this end, we define the bijective map
 \begin{align*} \varphi: \F_q^{m\times n} & \to \F_q^{m\cdot n}, \\ C & \mapsto c,
 \end{align*}
 which concatenates the rows of 
 $C = \begin{pmatrix}c_1^T,\ldots,c_m^T \end{pmatrix}^T$ into a single vector $c =(c_1,\ldots,c_m)$.
 Then, we refer to $c$ as the row vector form of $C$.
 The map $\varphi$ is extended to 
 \begin{align*} 
 \bar{\varphi}:  \left(\F_q^{m\times n}\right)^k &\to \F_q^{k\times(m\cdot n)}, \\  G_1,\ldots,G_k & \mapsto \bar{G}, \end{align*}
 by performing $\varphi$ matrix-wise on the input and stacking the resulting vectors. Then, it holds that
 $$\left\{ C= \sum\limits_{i=1}^k u_i G_i \mid u_i \in \mathbb{F}_q \right\} = \varphi^{-1}\left(\left\{ c = u \bar{G} \mid u \in \mathbb{F}_q^k, \bar{G}=\bar{\varphi}(G_1,\ldots,G_k) \right\}\right).$$
 Hence, $\bar{\varphi}$ allows for an unambiguous representation of the matrix code in a vector form.
 The vector form allows the introduction of an information set for the matrix code.
 \begin{definition}[Information Block] 
 Let $\mathcal{C} \subseteq \mathbb{F}_q^{m \times n}$ be a code  of dimension $k$ and $C\in\mathcal{C}$.
 Then, we refer to  $\mathcal{I}\times \mathcal{J}\subset \{1, \ldots, m\}\times\{1, \ldots, n\}$ of size $\tilde{m}\times\tilde{n}$ as an information block of $\mathcal{C}$, if 
 $$\lvert \{C_{\mathcal{I}\times\mathcal{J}} \mid C \in \mathcal{C} \} \rvert = \lvert \mathcal{C} \rvert.$$
 \end{definition}
 Note that not all codes of dimension $k$ have the same minimal size of an information block.
In particular, not every code has an information block of size $k$.
 \begin{example}
 Let $\mathcal{C} = \langle \begin{pmatrix} 1 & 0 \\ 0 & 0 \end{pmatrix}, \begin{pmatrix} 0 & 0 \\ 0 & 1 \end{pmatrix} \rangle \subseteq \mathbb{F}_2^{2 \times 2}$. 
 Then, $\mathcal{C}$ has an information block $\{1,2 \} \times \{1,2\}$ of size $4$ and no smaller information block exists.
 \end{example}
The definition of information block then also gives rise to a definition of systematic form for the generators.
\begin{definition}[Systematic Vector Form]\label{def:systematic_form}
Let $G_1, \ldots, G_k \in \mathbb{F}_q^{m \times n}$ be the generators of a code $\mathcal{C} \subset \mathbb{F}_q^{m \times n}$ of dimension $k$ and $\bar{G}$ the associated generator matrix in row vector form.
Let $\mathcal{I} \times \mathcal{J}$ be an information block for $\mathcal{C}$ of size $\tilde{m} \times \tilde{n}.$ 
Let $\tilde{G}_i\in\mathbb{F}_q^{\tilde{m}\times\tilde{n}}$  denote the submatrix formed from $G_i$ by considering only the rows of $\mathcal{I}$ and columns of $\mathcal{J}.$
We refer to the row vector form of $\tilde{G}_1,\ldots,\tilde{G}_k$ as $\tilde{G}\in\mathbb{F}_q^{k\times\tilde{m}\tilde{n}}$.
As $\mathcal{I} \times \mathcal{J}$ forms an information block, there exists a set $I \subset\{1, \ldots, \tilde{m}\tilde{n}\}$ of size $k$, such that $\tilde{G}_I$ has full rank.
 Then, the systematic row vector form of $G_1, \ldots, G_k$ is  given by 
 $$\bar{G}_s=\tilde{G}_I^{-1}\cdot\bar{G} \in \mathbb{F}_q^{k \times m n}.$$
\end{definition}
\section{Bounds in the Cover Metric and Density Results}\label{sec:bounds}
Observe that the Singleton bound in the cover metric is analogue to the Singleton bound for $\mathbb{F}_q$-linear rank-metric codes. 
\begin{theorem}[Cover Analogue of the Singleton Bound, Theorem 1, \cite{roth2}]
 Let $\mathcal{C}$ be an $[m \times n, k, d]$ cover metric code over $\mathbb{F}_q$, then
 $$k \leq \max\{n,m\}(\min\{n,m\}-d+1).$$
 \end{theorem}
 We denote by $F_C(n,m,q,w)$ the size of the cover-metric sphere, i.e.,
$$F_C(n,m,q,w)= \lvert \{ A \in \mathbb{F}_q^{m \times n} \ \mid \ \wt_C(A) =w\} \rvert$$
and similarly, we denote by $V_C(n,m,q,w)$  the size of the cover-metric ball, i.e.,
 $$V_C(n,m,q,w) = \lvert \{ A \in \mathbb{F}_q^{m \times n} \ \mid \ \wt_C(A) \leq w\} \rvert. $$
 The size of the cover ball is then 
 $$V_C(n,m,q,w)  = \sum_{i=0}^w F_C(n,m,q,i),$$
 which is needed for the analogue of the Gilbert-Varshamov bound in the cover metric. Unfortunately, there is no exact formula for the cover-metric sphere or ball. Instead, we need to work with a lower and upper bound on the size of the ball. 
 Recall that for $q\geq 2$ and $p \in [0,1]$ the $q$-ary entropy function is defined as
 $$H_q(p) = p\log_q(q-1)- p\log_q(p) -(1-p)\log_q(1-p).$$
 \begin{lemma}[\cite{liu2018list}]\label{lem:ballsize}
Let $n\leq m$ be positive integers and $q$ be a prime power. Let $0 < d \leq n.$
Then
$$ q^{md} \leq V_C(n,m,q,d) \leq (d+1) 2^{(m+n)H_2(d/(m+n))} q^{md}.$$
\end{lemma}
Let $A(q,m,n,d)$ denote the size of the largest code over $\mathbb{F}_q^{m \times n}$ of minimum cover distance $d$.
\begin{theorem}
Let $q$ be a prime power and $0<d \leq n \leq m$ be positive integers, then 
 $$ A(q,m,n,d) \geq \frac{q^{mn}}{V_C(n,m,q,d-1)}. $$
 \end{theorem}
In the case $n \leq m$ (thus if we let $n$ go to infinity, $m$ does as well), we have that 
$$\lim\limits_{n \to \infty} \frac{1}{mn} \log_q(V_C(n,m,q,d)) = \delta,$$ where $\delta$ denotes the relative minimum distance, i.e., $\delta=d/n.$  
Let us  consider the maximal information rate 
$$R(n,m,d,q) := \frac{1}{nm} \log_{q}(A(q,m,n,d)),$$
for $0 \leq d \leq n $. We define the relative minimum distance to be  $\delta := \frac{d}{n}.$
\begin{theorem}[Asymptotic Gilbert-Varshamov Bound]\label{asympt_GV}
For $n\leq m$, it holds that
$$\liminf\limits_{n \to \infty}R(n,m,\delta n, q) \geq \lim\limits_{n \to \infty} \left( 1 - \frac{1}{nm} \log_{q}(  V_C( n,m,q, \delta n)  ) \right) = 1-\delta.$$
\end{theorem}
Recall that in complexity theory we write $f(n) = \Omega(g(n))$, if $\limsup\limits_{n \to \infty} \left | \frac{f(n)}{g(n)} \right | >0$.
For example, $f(n) = \Omega(n)$ means that $f(n)$ grows at least polynomially in $n$.
 \begin{theorem}\label{thm:whp_GV}
  For every prime power $q$, $\delta \in [0,1), 0 \leq \varepsilon < 1-H_q(\delta)$ and sufficiently large positive integer $n\leq m$, the following holds for 
  $$k = \left\lceil \left(m - m\delta  +\frac{m}{n} -\frac{m+n}{n}H_2\left(\frac{\delta n-1}{m+n}\right)- \frac{\log_q(\delta n)}{n} -\varepsilon\right) n \right\rceil.$$
If we choose $G_1, \ldots, G_k \in \mathbb{F}_q^{m \times n}$ uniform at random, then the linear code generated by $G_1, \ldots, G_k$ has rate at least $1-\delta - \varepsilon$ and relative minimum cover distance at least $\delta$ with probability $1-e^{- \Omega(n)}$.  
 \end{theorem}
 \begin{proof}
 For the first statement on the rate, observe that $G_1, \ldots, G_k$ generate a code of dimension $k$, if 
 $$ \lambda_1 G_1 + \lambda_2 G_2 + \cdots + \lambda_k G_k =0,$$  with $\lambda_i \in \mathbb{F}_q$ implies that $\lambda_i=0$ for all $i \in \{1, \ldots, k\}.$ Thus, for $G_1$ we have $q^{nm}-1$ choices, that is all but the zero matrix, for $G_2$ we can choose any non-zero matrix which lies outside the span of $G_1$, thus we have $q^{nm}-q$ choices. Continuing in this manner, we get that the probability for $G_1, \ldots, G_k$ to generate a code of rate $\frac{k}{mn}$ is given by 
 $$\frac{ \prod\limits_{i=0}^{k-1} \left( q^{nm}-q^i\right)}{q^{mnk}} = \prod\limits_{i=0}^{k-1} \left(1-q^{i-mn}\right) \geq 1-e^{-\Omega(n)}.$$
  For the second statement, we first note that for any non-zero $x \in \mathbb{F}_q^k$, we have that 
 the codeword $$x_1 G_1 + \cdots + x_k G_k$$ is uniformly distributed in $\mathbb{F}_q^{m \times n}.$
 We now bound the counter probability, that is $$\Prob(\text{wt}_C(x_1 G_1 + \cdots + x_k G_k) < \delta n).$$ 
 This probability is given due to the uniform distribution of the codewords and Lemma~\ref{lem:ballsize} by
 \begin{align*}\frac{ V_C(n,m,q,\delta n-1)}{q^{mn}}  & \leq \frac{q^{m(\delta n-1)}2^{(m+n)H_2((\delta n-1)/(m+n))}\delta n}{q^{mn} } \\ &  \leq  q^{m\delta n+ (m+n)H_2((\delta n-1)/(m+n)) +\log_q(\delta n)-m -mn}.
 \end{align*}
 A union bound over all non-zero $x \in \mathbb{F}_q^k$, to get all non-zero codewords, now gives
\begin{align*} & q^k q^{m\delta n+ (m+n)H_2((\delta n-1)/(m+n)) +\log_q(\delta n)-m -mn} \\
 & = q^{ 1-n\varepsilon  } \leq e^{-\Omega(n)}.
\end{align*}
Thus, we get the claim.  
 \end{proof}
 \begin{definition}
 Let $n \leq m$ be positive integers.
 Let $\mathcal{C} \subseteq \F_q^{m \times n}$ of dimension $k$ and minimum cover distance $d$. We say that $\mathcal{C}$ is a maximum cover distance code, if 
 $$ d = n- \frac{k}{m} +1. $$
 \end{definition}
Since the asymptotic Singleton bound states that for any code of rate $R= \frac{k}{nm}$ and minimum distance $\delta=d/n$, the following must hold true:
$R \leq 1-\delta,$
we get the following corollary following Theorem \ref{thm:whp_GV}.
\begin{corollary}
Let $n \leq m$. The density of maximum cover distance codes in $\mathbb{F}_q^{m \times n}$ is 1, for $n$ going to infinity. 
\end{corollary}
In addition, we also get that the density is 1 if we let the field size grow. 
\begin{lemma}
 For fixed $k\leq n \leq m$, the density of maximum cover distance codes is 1, for $q$ going to infinity. 
\end{lemma}
\begin{proof}
    Let $\mathcal{C}$ be a random $[m \times n, k]$ cover-metric code and $\varepsilon>0$.
    Then, the probability that $\mathcal{C}$ has minimum distance smaller than $d=n-\frac{k+\varepsilon}{m}+1$ 
    can be bounded as
    $$\Prob(\exists C\in\mathcal{C}\setminus\{0\} \ \text{with} \  \text{wt}_C(C)\leq d-1) \leq \frac{d \cdot q^{m(d-1)} 2^{(m+n)H_2(d/(m+n))}}{q^{m n}},$$
    using the same arguments as in the proof of Theorem \ref{thm:whp_GV}.
    We define $$c=d \cdot2^{(m+n)H_2(d/(m+n))},$$ which is a constant in $q$ and simplify $m(d-1)-mn = -k-\varepsilon$, thus getting as an upper bound on the probability  $c q^{-k+\varepsilon}.$
    Then, the claim follows from the union bound, as
    $$\Prob(\forall \ C \in \mathcal{C}\setminus\{0\} \ \text{we have} \ \text{wt}_C(C)\geq d) = 1-c\cdot q^{-\varepsilon},$$ which tends to 1 for $q$ going to infinity.  
\end{proof}

 \section{Cover Decoding Problem}\label{sec:NP}
 Although NP-hardness is in theory only defined for decisional problems, their computational counterparts are usually also called NP-hard. In fact,  if one could solve the computational problem one would also be able to answer the decisional one. However, in the following we will only consider the decisional version of the problems. 
 The (decisional) cover-metric decoding problem asks if there exists an error matrix, for a given erroneous codeword of a random code.
 \begin{problem}[Cover Decoding Problem]\label{CDP}
 Given $G_1, \ldots, G_k \in \mathbb{F}_q^{m \times n}$, a positive integer $t\leq \min\{m,n\}$ and $Y \in \mathbb{F}_q^{m \times n}$, does there exist a matrix $E \in \mathbb{F}_q^{m \times n}$, such that  there exist $x_1, \ldots, x_k \in \mathbb{F}_q$ with
 $$x_1G_1 + \cdots + x_kG_k +E = Y,$$ and $\text{wt}_C(E) \leq t$?
 \end{problem}
 This problem is equivalent to the following  problem, in the same manner as this equivalence holds for other metrics and ambient spaces. 
  \begin{problem}[Cover Syndrome Decoding Problem]\label{CSDP}
 Given $H_1, \ldots, H_{mn-k} \in \mathbb{F}_q^{m \times n}$, $s=(s_1, \ldots, s_{mn-k})\in \mathbb{F}_q^{mn-k}$, a positive integer $t\leq \min\{m,n\}$, does there exist an $E \in \mathbb{F}_q^{m \times n}$, such that $$\langle H_i, E \rangle = s_i$$ for all $i \in \{1, \ldots, mn-k\}$ and $\text{wt}_C(E) \leq t$?
 \end{problem}
  Recall that an algorithm that solves the lowest-weight codeword problem can also solve the decoding problem, by adding $Y$ as a generator to the random code. The error matrix $E$  is now the codeword of the new code of smallest cover weight. Thus, there is a reduction from the  decoding problem to the lowest-weight codeword problem. 
  \begin{problem}[Lowest Cover Weight Codeword Problem]\label{LCWCP}
 Given a positive integer $t\leq \min\{m,n\}$ and $G_1, \ldots, G_k \in \mathbb{F}_q^{m \times n}$, does there exist $(x_1, \ldots, x_k) \in \mathbb{F}_q^k$ such that
 $$\text{wt}_C(x_1G_1 + \cdots + x_kG_k) \leq t?$$
 \end{problem}
 We now show that Problem \ref{LCWCP} is NP-hard, by giving a polynomial time reduction from the Lowest Hamming Weight Codeword Problem.
   \begin{problem}[Lowest Hamming Weight Codeword Problem]
  Given $G \in \mathbb{F}_q^{k \times n}$, a positive integer $t\leq n$, does there exist $x \in \mathbb{F}_q^k\setminus\{0\}$ such that
 $$\text{wt}_H(xG) \leq t?$$
 \end{problem}
 \begin{theorem}\label{theo:coverIsNP}
  The Lowest Cover Weight Codeword Problem is NP-complete.
 \end{theorem}
\begin{proof}
Let $(q,n,G,t)$ be a random instance of the Lowest Hamming Weight Codeword Problem, which is known to be NP-complete \cite{berlekamp1978inherent, barg}.
 Let $g_i$ denote the $i$-th row of $G$.
 We define $G_i\in \mathbb{F}_q^{(t+1)\times n}$ as the matrix which results from stacking $g_i$ $t+1$ times, i.e., $G_i = (g_i^\top,\ldots,g_i^\top)^\top$.
 Then, $(q,G_1, \ldots, G_k,t)$ is an instance of the Lowest Cover Weight Codeword Problem with $m=t+1$.
 Due to the method for constructing the generators, this instance of the Lowest Cover Weight Codeword Problem has the property that all $t+1$ rows of a codeword $A=\sum_{i=1}^{k} u_i G_i$ are identical.
 We refer to any such row of $A$ as $a$.
 Solving this particular  problem also solves the  problem in the Hamming metric.
 
 We begin by considering the case in which the answer to the cover-metric problem is "yes", i.e, there is a codeword $A=\sum_{i=1}^{k} u_i G_i$ with  $(u_1,\ldots,u_k)\neq 0$ and cover weight $\leq t$.
 Due to the special choice of the $G_i$, any cover of weight $\leq t$ has to be formed by columns only.
 We conclude that $a$ has at most $t$ non-zero entries.
 As $a=(u_1,\ldots,u_k) G$, $a$ is also a solution to the Hamming-metric problem.
 
 On the other hand, if the answer to the cover-weight problem is "no", every matrix $A=\sum_{i=1}^{k} u_i G_i$  with  $(u_1,\ldots,u_k)\neq 0$ has cover weight $>t$.
 As the rows of $A$ are identical, also each row $a$ has to have at least $t+1$ non-zero entries.
 Therefore, each $a=(u_1,\ldots,u_k) G$ with $(u_1,\ldots,u_k)\neq 0$ has to have at least $t+1$ non-zero entries and the answer to the Hamming-weight problem is also "no".
 
 This proves the NP-hardness of the Lowest Cover Weight Codeword Problem.
 Finally, we note that the problem lives in NP, since any candidate solution can be verified in polynomial time.  
\end{proof}
Similarly, we can also show that Problem \ref{CDP} is NP-hard, by giving a polynomial-time reduction from the Hamming Decoding Problem.
 \begin{problem}[Hamming Decoding Problem]
  Given $G \in \mathbb{F}_q^{k \times n}$, $r \in \mathbb{F}_q^n$, a positive integer $t\leq n$, does there exist $e \in \mathbb{F}_q^n$ such that $r-e \in \langle G \rangle$ and 
 $$\text{wt}_H(e) \leq t.$$
 \end{problem}
\begin{theorem}\label{theo:coverdIsNP}
  The  Cover  Decoding Problem is NP-complete.
 \end{theorem}
\begin{proof}
Let $(q,G,r,t)$ be a random instance of the  Hamming  Decoding Problem, which is known to be NP-complete \cite{berlekamp1978inherent, barg}.
 Let $g_i$ denote the $i$-th row of $G$.
 We define $G_i\in \mathbb{F}_q^{(t+1)\times n}$ as the matrix which results from stacking $g_i$ $t+1$ times and $R \in \mathbb{F}_q^{(t+1) \times n}$ the matrix  which results from stacking $r$ $t+1$ times.
 Then, $(q,G_1, \ldots, G_k,R, t)$ is an instance of the  Cover  Decoding Problem with $m=t+1$.
Similar to before, we have that all $t+1$ rows of a codeword $C=\sum_{i=1}^{k} u_i G_i$ are identical and solving this particular  problem also solves the  problem in the Hamming metric.
 Let us assume that the answer to the cover-metric problem is "yes", i.e, there is a matrix $E \in \mathbb{F}_q^{(t+1) \times n}$ such that $R-E \in \langle G_1, \ldots, G_k)$ and $\text{wt}_C(E) \leq t$. Since all $t+1$ rows of $R$ are identical and all rows of a codeword in $\langle G_1, \ldots, G_k \rangle$ are identical, we must get that all rows of $E$ are identical as well. Let us denote such a row by $e \in \mathbb{F}_q^n.$
 Since $E$ consists of $t+1$ identical rows, any cover of weight $\leq t$ can be formed by columns only, thus $\text{wt}_H(e) \leq t.$ Since $e$ is such that $r-e \in \langle G \rangle$ by construction, we have that the answer to the Hamming-metric problem is also 'yes'. 

 On the other hand, if the answer to the cover-weight problem is "no", every matrix $E \in \mathbb{F}_q^{(t+1) \times n}$ which is such that $R-E  \in \langle G_1, \ldots, G_k\rangle$  has cover weight $>t$.
 As the rows of $E$ are identical, also each row $e$, which are all the vectors with $r-e \in \langle G \rangle$, has $\text{wt}_H(e)>t$.
 Thus, the answer to the Hamming-weight problem is also "no".
 Since this problem clearly lives in NP, we get the claim.  
\end{proof}
\section{Generic Decoding in the Cover Metric}\label{sec:generic}
In this section, we introduce and compare different approaches to generic decoding in the cover metric.

In order to avoid more technical statements, we fix that $m\cdot R \leq n\leq m$.
A generalization to $n < R\cdot m$ is straightforward but does not provide any further insights.

\subsection{Error Model}

In order to evaluate the complexity of the decoding approaches, it is essential to specify the underlying error model.
The most general one is given as follows.

\begin{definition}[General error model]\label{def:general_err}
An error $E$ of cover weight $t$ is created by choosing $E$ at random from 
$\left\{A\in\mathbb{F}_q^{m \times n} \mid \wt_C(A)=t\right\}$.
 \end{definition}
  Another error model \cite{roth} allows a simplified generation of error matrices, as well as a simplified analysis of the complexity of generic decoding: 
\begin{definition}[Simple error model]\label{def:simple_err}
 An error $E$ of cover weight $t$ is created as follows:
 first pick $t$ of the $n+m$ rows and columns  at random.
 Then, create $E$ by setting the elements in the picked rows and columns to random elements of $\mathbb{F}_q$.
 All remaining elements are zero.
 The cover weight of $E$ is checked. If $\wt_C(E)<t$, the process is repeated.
 \end{definition}

The error models are similar, but not identical.
In the general error model, every matrix in the sphere of radius $t$ is chosen with equal probability
$$\Prob(E\mid\text{general model}) = \frac{1}{F_C(n,m,q,t)}\quad \forall E\in\{A \in \mathbb{F}_q^{m \times n} \ \mid \ \wt_C(A) =t\}.$$
This is not the case for the simple error model, for which some matrices are more likely than others.
Define $J_{t,n}\in \mathbb{F}_q^{t \times n}$ as the all-one matrix.
Consider the two error matrices
$$E_1 = 
\begin{pmatrix}
\text{Id}_t & 0 \\
 0 & 0
\end{pmatrix}
\text{ and }
E_2 = 
\begin{pmatrix}
J_{t,n} \\
0 
\end{pmatrix},$$
which both have cover weight $t$.
Then, it holds that
$$\Prob(E_1\mid\text{simple model}) = 2^t  \Prob(E_2\mid\text{simple model}).$$
The matrix $E_2$ is generated by choosing exactly the first $t$ rows and the all-one vectors for each line, whereas $E_1$ is generated by $2^t$ combinations of lines and choosing one particular vector for each line.

Nevertheless, the simple model can be used to accurately approximate the general error model in practice, as the following theorem shows.

\begin{theorem}[Unique Minimal Cover]\label{theo:unique_cover}
 Let $E\in\F_q^{m\times n}$ be a matrix which is created by choosing $t\leq\min\{m,n\}$ out of $m+n$ lines at random and filling the chosen rows and columns with random elements from $\mathbb{F}_q$.
 Then, for sufficiently large $n$, $E$ has a unique minimal cover of cardinality $t$ with high probability.
\end{theorem}

 \begin{proof}
Let $t_x$ denote the number of chosen rows, $t_y$ the number of chosen columns and $t=t_x+t_y$ the total number of chosen lines.
As the cover weight is invariant under the permutation of rows and columns, we can assume without loss of generality that $E$ is of the form
$$E= 
\begin{pmatrix}
A & B \\
C & 0
\end{pmatrix},$$
where $A\in \F_q^{t_x\times t_y}$, $B\in \F_q^{t_x\times (n-t_y)}$ and $C\in \F_q^{(m-t_x)\times t_y}$.
In total, there are 
$$q^{t_x\cdot t_y}q^{t_x\cdot (n-t_y)}q^{(m-t_x)\cdot t_y}$$
possibilities for choosing such a matrix $E$.
In the following, we develop a lower bound for the number of choices which lead to a matrix $E$, which has exactly one minimal cover of size $t$. 
In other words, we bound the number of choices which lead to a matrix of cover weight $t$ which cannot be created using any other set of lines.
For this, the only possibility to cover $B$ has to be the cover consisting of all $t_x$ rows.
This is the case if and only if the union of the supports of any $\ell$ rows has cardinality at least $\ell+1$ for $\ell \in\{1,\ldots,t_x\}$.
If there is a union of the supports of $\ell$ rows, which has cardinality at most $\ell$, one can replace those rows in the cover by the columns which are indexed by the union.
In particular, none of the rows of $B$ is allowed to have Hamming weight zero or one.
In order to bound the number of possible matrices $B$, one can assume that the rows are chosen sequentially such that the required property is satisfied.
We consider the number of possibilities for choosing the $t_x$-th row.
This number is minimized if the previous $t_x-1$ rows are chosen in the most restricting way, which is the case if the union of the support of any $\ell$ rows has size exactly $\ell+1$.
Up to a permutation of columns, the only constellation of the previous $t_x$ rows, which satisfies this is given by 
$$\begin{pmatrix}
J_{t_x-1,1} & \text{Id}_{t_x-1} & 0 
\end{pmatrix}\in\F_q^{(t_x-1)\times(n-t_y)}$$
Then, out of the total number of  $q^{n-t_y}$ possibilities for  choosing the last row, we have to discard those, which have zeros in the last $n-t_x-t_y$ last positions.
The discarded possibilities include the weight zero vector and the weight one vectors, for which the non-zero entry is located in the first $t_x$ positions.
Further, we discard the weight one vectors with nonzero entry in the last $n-t_y-t_x=n-t$ positions, of which there are $(q-1)(n-t)$.
Hence, there are always at least 
$$q^{n-t_y} - (q-1)(n-t) - q^{t_x}$$
possibilities to choose the $t_x$-th row.
For the previous rows there are less restrictions. 
Therefore, we can conclude that there are at least
$$\left(q^{n-t_y} - (q-1)(n-t) - q^{t_x}\right)^{t_x}$$
possibilities for choosing $B$.
Using an analogous argument, one can show that there are at least 
$$\left(q^{m-t_x} - (q-1)(m-t) - q^{t_y}\right)^{t_y}$$
possibilities for choosing $C$.
Finally, there are $q^{t_x\cdot t_y}$ possibilities for choosing $A$ and all of them are suitable for obtaining a matrix $E$ with a unique minimal cover of cardinality $t$.
In summary, the number of suitable possibilities is at least 
$$ q^{t_x\cdot t_y} \left(q^{n-t_y} - (q-1)(n-t) - q^{t_x}\right)^{t_x} 
\left(q^{m-t_x} - (q-1)(m-t) - q^{t_y}\right)^{t_y}.$$
Therefore, the probability of creating such a matrix $E$ with unique minimal cover of size $t$ is bound from below by
\begin{multline*}
\frac{q^{t_x\cdot t_y} \left(q^{n-t_y} - (q-1)(n-t) - q^{t_x}\right)^{t_x} \left(q^{m-t_x} - (q-1)(m-t) - q^{t_y}\right)^{t_y}}{q^{t_x\cdot t_y}q^{t_x\cdot (n-t_y)}q^{(m-t_x)\cdot t_y}} \\
= \left(1-\frac{(q-1)(n-t)}{q^{n-t_y}}-q^{t-n}\right)^{t_x}\left(1-\frac{(q-1)(m-t)}{q^{m-t_x}}-q^{t-m}\right)^{t_y},
\end{multline*}
which tends to $1$ for $n\leq m$, $t=n(1-R)/2$ and $n$ going to infinity. 
\end{proof}

From Theorem $\ref{theo:unique_cover}$ we get a lower bound  on the size of the spheres, which is of a similar form as the upper bound given in \cite{martinez2022multilayer}.
\begin{corollary}
Let $t \leq \min\{n, m\} $ be positive integers. Then,
\begin{align*} 
F_C(n,m,q,t)\geq \sum_{t_x+t_y=t}\binom{m}{t_x}\binom{n}{t_y}q^{t_x\cdot t_y} \left(q^{n-t_y} - (q-1)(n-t) - q^{t_x}\right)^{t_x} \\ 
\cdot \left(q^{m-t_x} - (q-1)(m-t) - q^{t_y}\right)^{t_y}.
\end{align*}
\end{corollary}

 \subsection{Prange-like Decoding in the Cover Metric}
 The classical algorithm in the Hamming metric by Prange, formulated via the generator matrix is such that, given a received word $r$, a generator matrix $G$ and a weight $t$, one chooses an information set $I$ and checks whether $r-r_IG_I^{-1}G$ has weight $t$. If the error vector has support outside the information set, this procedure succeeds. 
 
 In the following, we introduce and analyze a simple Prange-like generic decoder in the cover metric. 
 For this, we  focus on the case $ mR<n\leq m$.
 Let $\mathcal{C}=\left\{\sum_{i=0}^{k-1}u_iG_i\mid u_i\in\mathbb{F}_q\right\}$ with $G_i\in\mathbb{F}_q^{m\times n}$.
 In order to recover $C\in\mathcal{C}$ from $R=C+E$ with $\wt_C(E)=t$, one can proceed as follows:

\begin{algorithm}[h!]
\caption{Cover-Metric Prange Algorithm}\label{algo:prange}
\begin{flushleft}
Input: The generators $G_1, \ldots, G_k \in \mathbb{F}_q^{m \times n}$ of the code $\mathcal{C}$, the received matrix $R \in \mathbb{F}_q^{m \times n}$ and the cover weight $t \in \mathbb{N}$. \\ 
Output: $C \in \mathcal{C}$ with $\text{wt}_C(R-C) \leq t$.
\end{flushleft}
\begin{algorithmic}[1]
\State Choose $\hat{t}_x \leq m$ and set $\hat{t}_y = \left\lfloor n-\frac{k}{m-\hat{t}_x}\right\rfloor= n-\left\lceil \frac{k}{m-\hat{t}_x}\right\rceil$.
\State Choose  $\mathcal{I} \times \mathcal{J} \subset \{1, \ldots, m\} \times \{1, \ldots, n\}$ of size $(m-\hat{t}_x) \times (n-\hat{t}_y) =\tilde{m} \times \tilde{n}$ at random.
\State Compute $\bar{\varphi}(G_1, \ldots, G_k)= \bar{G}.$
\State Compute $\bar{\varphi}((G_1)_{\mathcal{I} \times \mathcal{J}}, \ldots, (G_k)_{\mathcal{I} \times \mathcal{J}})=\tilde{G} \in \mathbb{F}_q^{k \times  \tilde{m}\tilde{n}}$.
\State Choose a set $I \subset \{1, \ldots, \tilde{m}\tilde{n}\}$ of size $k$, such that $\text{rk}(\tilde{G}_I)=k.$
\State Compute $\bar{G}_s=\tilde{G}_I^{-1}\bar{G} \in \mathbb{F}_q^{k \times mn}$. 
\State  Compute  $r=\varphi(R_{\mathcal{I} \times \mathcal{J}})\in\F_q^{(m-\hat{t}_x)(n-\hat{t}_y)}$.
\State Compute $\hat{c}=r_I\bar{G}_s$. 
\If{ $\text{wt}_C(R- \varphi^{-1}(\hat{c})) \leq t$}{ Return $\hat{C} = \varphi^{-1}(\hat{c}).$}
\EndIf
\State Else, return  to Step 2.
\end{algorithmic}
\end{algorithm}
Similar to the classical case, we have that one iteration succeeds if the chosen information-block is such that $\text{wt}_C(E_{\mathcal{I} \times \mathcal{J}})=0.$
 \begin{remark}
 The number of erased columns in Step (2) is chosen such that $(m-t_x)(n-t_y)\geq k$.
 This guarantees that we obtain an information block with high probability. In fact, let $\tilde{m}$ and $\tilde{n}$ be integers such that $\tilde{m}\tilde{n}\geq k$.
 Further, let $\tilde{G}_i\in\mathbb{F}_q^{\tilde{m}\times\tilde{n}}$ for $i \in \{1,\ldots,k\}$ be chosen at random.
 Define $\tilde{\mathcal{C}}=\left\{u\cdot \bar{\varphi}(\tilde{G}_1,\ldots,\tilde{G}_k)\mid u\in\mathbb{F}_q^k\right\}$.
 Then,  
 $$
 \Prob(|\tilde{\mathcal{C}}|=q^k) = \Prob(\text{rk}(\bar{\varphi}(\tilde{G}_1,\ldots,\tilde{G}_k)=k)  = \prod_{s=0}^{k-1} \left(1-q^{s-\tilde{m}\tilde{n}}\right).
 $$
 This probability is  at least as high as the probability that $k$ random positions of a random code form an information set.
 We conclude that the probability that a random block of size $\tilde{m}\tilde{n}\geq k$ is an information block is, for most codes, a constant in the same order of magnitude as $1$ \cite{horlemann2021information}.
 \end{remark}
 Now, we analyze the average success probability of the Prange-like decoder and, in doing so, determine according to which probability mass function the $\hat{t}_x$ should be chosen in order to obtain the maximum average success probability.
 Let $t_x$ be the number of rows in the minimum-size cover of $E$ and $t_y$ the number of columns.
 According to the definition of the cover weight, we have $\wt_C(E) = t = t_x +t_y$.
  Except for a negligible number of cases, the Prange-like decoder only succeeds if indeed all rows and columns of $E$, which have non-zero weight, are erased.
 Therefore, the probability of the event that $\hat{C}$ is equal to $C$ given $t_x$ and $\hat{t}_x$ is tightly approximated by
 \begin{equation*}
  \Prob(S\mid t_x,\hat{t}_x) =  \frac{\binom{\hat{t}_x}{t_x}}{\binom{m}{t_x}}\cdot \frac{\binom{\hat{t}_y}{t_y}}{\binom{n}{t_y}}.
 \end{equation*}
 In order to enable a simple complexity analysis of the Prange-like decoder, we assume the simple error model introduced in Definition \ref{def:simple_err}, i.e., the $t$ rows and columns in $E$ which contain non-zero elements are picked at random out of the $m+n$ available rows and columns.
 Theorem \ref{theo:unique_cover} proves that  a matrix $E$, which is generated in this way, is rarely rejected.
 Therefore, one can conclude that the distribution of $t_x$ is well approximated by the hypergeometric distribution, which means that its probability mass function is given by 
 \begin{equation*}
  \Prob(t_x\mid t) = \frac{\binom{m}{t_x}\binom{n}{t-t_x}}{\binom{m+n}{t}} =  \frac{\binom{m}{t_x}\binom{n}{t_y}}{\binom{m+n}{t}}.
 \end{equation*}
 Using the law of total probability, the success probability given $\hat{t}_x$ is calculated as
 \begin{equation*}
  \Prob(S\mid \hat{t}_x) = \sum_{t_x = 0}^{t} \Prob(S\mid \hat{t}_x, t_x) \Prob(t_x\mid t) = \sum_{t_x = 0}^{t} \frac{\binom{\hat{t}_x}{t_x}\binom{\hat{t}_y}{t-t_x}}{\binom{m+n}{t}}.
 \end{equation*}
 Applying Vandermonde's identity, the expression is simplified to 
  \begin{equation*}
  \Prob(S\mid \hat{t}_x) = \frac{\binom{\hat{t}_x+\hat{t}_y}{t}}{\binom{m+n}{t}}.
 \end{equation*}
 Hence, the overall success probability is given by
 \begin{equation*}
  \Prob(S) = \binom{m+n}{t}^{-1}\cdot\sum_{\hat{t}_x = 0}^{t} \binom{\hat{t}_x+\hat{t}_y}{t}\cdot \Prob(\hat{t}_x),
 \end{equation*}
 where $ \Prob(\hat{t}_x)$ should be chosen such that $\Prob(S)$ is maximized as this minimizes the average number of required iterations.
 Formally, this optimization problem can be stated as 
 \begin{equation*}
  \Prob(\hat{t}_x) = \argmax_{\substack{p_{\hat{t}_x}: \sum_{\hat{t}_x} p_{\hat{t}_x}=1\\p_{\hat{t}_x}\geq0 }} \ \sum_{\hat{t}_x = 0}^{t} \binom{\hat{t}_x+\hat{t}_y}{t}\cdot p_{\hat{t}_x}.
 \end{equation*}
 Using a standard Lagrange multiplier argument or some intuition, one finds that one solution is always given by
 \begin{equation*}
  \Prob(\hat{t}_x) = \begin{cases}
                    1, & \,  \text{if $\hat{t}_x\in\{0,\ldots,t\}$ is one particular $\hat{t}_x$ which maximizes $\binom{\hat{t}_x+\hat{t}_y}{t}$,} \\
                    0, & \, \text{else.}
                 \end{cases}
 \end{equation*}
 Let $t^*_x$ be the chosen $\hat{t}_x\in\{0,\ldots,t\}$ which maximizes $\binom{\hat{t}_x+\hat{t}_y}{t}$, i.e., $\hat{t}_x+\hat{t}_y$. 
 Then, due to $\hat{t}_y=n-\left\lceil\frac{k}{m-\hat{t}_x}\right\rceil$, the success probability of the Prange-like decoder is given by
 \begin{equation}\label{eq:succ_prob_prange_like}
  \Prob(S) = \frac{\binom{\left\lfloor t^*_x+n-\frac{k}{m-t^*_x}\right\rfloor}{t}}{\binom{m+n}{t}}
 \end{equation}
 In order to bound the probability given in \eqref{eq:succ_prob_prange_like} from above, one can relax the requirement that $t^*_x$ has to be an integer.
 Let $f(x)$ be the function 
 \begin{align*}
  f(x)\colon [0, m] &\to \mathbb{R}, \\  x & \mapsto x+n-\frac{k}{m-x}.
 \end{align*}
 Showing that $f(x)$ attains its maximum value at $x^*=m-\sqrt{k}$ is straightforward.
 It follows that 
 \begin{equation*}
  \left\lfloor t^*_x+n-\frac{k}{m-t^*_x}\right\rfloor \leq \left\lfloor f(x^*)\right\rfloor = \left\lfloor m+n-2\sqrt{k}\right\rfloor = m+n-\left\lceil2\sqrt{k}\right\rceil.
 \end{equation*}
 In order to bound the probability given in \eqref{eq:succ_prob_prange_like} from below, one can choose $\hat{t}_x = m-\left\lceil\sqrt{k}\right\rceil$.
 Note that for this choice of $\hat{t}_x$, $ mR < n \leq m$ ensures that  $\hat{t}_y = n-k/\lceil\sqrt{k}\rceil\approx n-\sqrt{k} \geq 0$.

 Then, due to 
 \begin{equation*}
\left\lceil\frac{k}{\left\lceil\sqrt{k}\right\rceil} + \left\lceil\sqrt{k}\right\rceil\right\rceil
\leq \left\lceil\frac{k+(\sqrt{k}+1)^2}{\sqrt{k}}\right\rceil
\leq \left\lceil2\sqrt{k}\right\rceil+3,  
 \end{equation*}
it holds that
 \begin{equation*}
  \left\lfloor \hat{t}_x+n-\frac{k}{m-\hat{t}_x}\right\rfloor  =  \left\lfloor m-\left\lceil\sqrt{k}\right\rceil+n-\frac{k}{\left\lceil\sqrt{k}\right\rceil}\right\rfloor \geq m+n -\left\lceil2\sqrt{k}\right\rceil-3.
\end{equation*}
In summary, the success probability can be tightly bounded as
\begin{equation*}
 \binom{ m+n -\left\lceil2\sqrt{k}\right\rceil-3}{t} \leq  \Prob(S) \cdot\binom{m+n}{t} \leq \binom{ m+n -\left\lceil2\sqrt{k}\right\rceil}{t}.
\end{equation*}
It is worth noting that the upper bound on the success probability is similar to the formula for the success probability of Prange's algorithm in the Hamming metric, i.e.,
\begin{equation*}
 \Prob(S_\mathrm{Prange}) = \binom{n-k}{t}\binom{n}{t}^{-1}.
\end{equation*}
Let $\tau = t/n$ and $H_q(p)$ denote the $q$-ary entropy function.
Then, asymptotically, Prange's algorithm in the Hamming metric requires on average $2^{c_\mathrm{Prange}(R,\tau)\cdot n}$ iterations, where
\begin{equation*}
 c_\mathrm{Prange}(R,\tau) = H_2\left(\tau \right)-(1-R)\cdot H_2\left(\tau/(1-R) \right).
\end{equation*}
For $\tau_\mathrm{GV} = \frac{d_\mathrm{GV}-1}{2n}$, the constant is given by
\begin{equation*}
 c_\mathrm{PrangeGV}(R) = c_\mathrm{Prange}(R,\tau_\mathrm{GV}) = H_2\left( H_q^{-1}(1-R)/2\right) - (1-R)\cdot H_2\mleft(\frac{H_q^{-1}(1-R)}{2(1-R)}\mright).
\end{equation*}

A similar analysis can be carried out for the proposed Prange-like decoder in the cover metric.
 \begin{theorem}[Asymptotic Complexity of the Prange-like Decoder]\label{asymp_prange}
  Let $m = n$, $\tau = t/n$ and $R=k/(nm)$.
  Then, the number of iteration is asymptotically in the order of $2^{c_\mathrm{cover}(R,\tau)\cdot(n+m)}$
  with
  \begin{equation}\label{eq:complexity_c_cover}
   c_\mathrm{cover}(R,\tau) = H_2\mleft(\frac{\tau}{2}\mright) - \left(1-\sqrt{R}\right)\cdot H_2\mleft(\frac{\tau}{2\cdot(1-\sqrt{R})}\mright),
  \end{equation}
  where  $H_2(p)$ denotes the binary entropy function.
  For $\tau_\mathrm{coverGV} = \frac{d_\mathrm{coverGV}-1}{2n}$, the constant simplifies to
  \begin{equation*}
    c_\mathrm{coverGV}(R) 
    = c_\mathrm{cover}(R,\tau_\mathrm{coverGV}) 
    = H_2\mleft(\frac{1-R}{4}\mright) -\left(1-\sqrt{R}\right) H_2\mleft(\frac{1+\sqrt{R}}{4}\mright).
  \end{equation*}
    The cost of each iteration is dominated by the required Gaussian elimination, i.e., it is in the order of $n^3$.
 \end{theorem}
 \begin{proof}
 The average number of required iterations is given by $\Prob(S)^{-1}$.
 For $m=n \rightarrow \infty$, the lower bound and the upper bound on $\Prob(S)$ coincide, i.e.,
 \begin{equation*}
  O\left(\Prob(S)^{-1}\right) = O\left(\binom{n+m}{t}\binom{n+m-2\sqrt{R mn}}{t}^{-1}\right)
 \end{equation*}
 As we consider the case $m=n$, it holds that $\sqrt{m\cdot n}=\frac{m+n}{2}$,  which simplifies the previous equation to
 \begin{equation*}
 O\left(\Prob(S)^{-1}\right) = O\left(\binom{n+m}{t}\binom{(n+m)(1-\sqrt{R})}{t}^{-1}\right).
 \end{equation*}
 Further, a well-known formula for binomial coefficients states that $O\left(\binom{k}{j}\right)=\tilde{O}\left(2^{H_2(j/k)\cdot k}\right)$.
 Using this relation, one observes that
 \begin{equation*}
  O(\Prob(S)^{-1}) = \tilde{O}\mleft(
  2^{(m+n)H_2\mleft(\frac{t}{n+m}\mright)} \cdot 2^{-(m+n)(1-\sqrt{R})H_2\mleft(\frac{t}{(n+m)(1-\sqrt{R})}\mright)}
  \mright),
 \end{equation*}
 which proves \eqref{eq:complexity_c_cover}.
 In the cover metric, the Gilbert-Varshamov bound states that 
 \begin{equation*}
  \lim_{n \to \infty}\tau_\mathrm{coverGV} = \lim_{n \to \infty}\frac{\left\lfloor (d_\mathrm{coverGV}-1)/2\right\rfloor}{n} = \lim_{n \to \infty} \frac{n(1-R)}{2n} = \frac{1-R}{2}.
 \end{equation*}
Thus, we get the claim.  
 \end{proof}

\begin{remark}
It is worth noting that $O(\Prob(S)) = \tilde{O}(2^{(n+m)\cdot c_\mathrm{coverGV}(R) })$ implies
$$\lim\limits_{n=m \to \infty}\frac{\log_2\left(\Prob(S)^{-1}\right)}{n\cdot m} = 0.$$
This behaviour is different from generic decoding in other metrics.
In the Hamming metric, generic decoding of a random binary code of length $n\cdot m$ succeeds with a probability, for which it holds that
$$\lim\limits_{n=m \to \infty}\frac{\log_2\left(\Prob(S_H)^{-1}\right)}{n\cdot m} = c_H(R),$$
where $c_H(R)$ is a constant which varies only slightly between the decoding algorithms.
\end{remark}

 In the following we consider a possible Stern-like adaption of our algorithm.
 However, similar to generic rank-metric decoders, we observe that the list sizes involved in such an algorithm would be too large to give an improvement on the cost. 
 The idea of Stern's algorithm in the Hamming metric is to split the error vector in the information set into two parts. We then assume that the Hamming weight in each part is $v/2$ and go through all lists of such error vectors solving the partial syndrome equations. The list sizes are then one of the factors in the cost of one iteration. 
 One can do a similar approach in the cover metric, by building the list $\mathcal{L}_1$ consisting of all matrices, where we have $\sqrt{k+\ell}$ rows of cover weight $v/2$ and the list $\mathcal{L}_2$ consisting of all matrices with $\sqrt{k+\ell}$ columns of cover weight $v/2$. Both lists then have size 
 $$\binom{\sqrt{k+\ell}}{v/2} \left(q^{\sqrt{k+\ell}}-1 \right)^{v/2}.$$
 Assuming as in Theorem \ref{asymp_prange}, that $m=n$, we get  $\sqrt{k+\ell}=\sqrt{R  n^2+\ell} \geq \sqrt{R}n$ and  can bound 
 $$\binom{\sqrt{k+\ell}}{v/2}   q^{\frac{\sqrt{k+\ell}  v}{2}} \geq \binom{\sqrt{R}n}{v/2}   q^{\frac{n \sqrt{R}  v}{2}},$$
 which implies for large $n$ that
 $$ \log_2\mleft( \binom{\sqrt{k+\ell}}{v/2}   q^{\frac{\sqrt{k+\ell}  v}{2}} \mright)> \frac{n \sqrt{R}  v}{2}\log_2(q).$$
 We observe that the size is considerably increased in comparison to the Hamming metric, especially for $q>2$.
 
 One can verify that the list size is larger than the cost of the Prange-like decoder for relevant rates.
 Thus, we will not decrease the cost of the Prange-like decoder in the cover metric when following this idea. 

 \subsection{Generic Decoding based on the Rank Metric}
 
The Lowest Cover Weight Codeword Problem also has a clear connection to the well-known MinRank Problem, which has been shown to be NP-complete in \cite{minrank}.
 
 \begin{problem}[MinRank Problem]\label{minrank}
  Given $G_1, \ldots, G_k \in \mathbb{F}_q^{m \times n}$, a positive integer $t\leq \min\{m,n\}$, does there exist $(x_1, \ldots, x_k) \in \mathbb{F}_q^k$ such that
 $$\text{wt}_R(x_1G_1 + \cdots + x_kG_k) \leq t?$$
 \end{problem}
Observe  that for any $A \in \mathbb{F}_q^{m \times n}$ we have 
$\text{wt}_R(A) \leq \text{wt}_C(A).$
 If one can solve the Lowest Cover Weight Codeword Problem and find  a solution $x_1, \ldots, x_k \in \mathbb{F}_q$, this then also forms a solution to the MinRank problem as
 $$\text{wt}_R(x_1 G_1 + \cdots +x_kG_k) \leq \text{wt}_C(x_1 G_1 + \cdots +x_kG_k)\leq t.$$
 This will be useful to compare our decoding algorithm with the already known algorithms for the MinRank problem \cite{bardet}. In fact, giving such MinRank solver the inputs $(G_1, \ldots, G_k,t)$ the solver will output a list of solutions $\mathcal{L}=\{x_1, \ldots, x_N\}$. Since we assume that the generating matrices $G_i \in \mathbb{F}_q^{m \times n}$ are chosen uniformly at random, the expected number of solutions to the MinRank problem is $$N=\frac{q^{t(m+n-t)}}{q^{mn-k}}.$$ 
 Observe that for $t= n(1-R)/2$ and $m=n$, we get
 $$N=\frac{q^{t(m+n-t)}}{q^{mn-k}}=q^{-n^2(1-R)^2/4},$$ which tends to 0 for $n$ going to infinity. 
 This is due to 
 $$\lvert\mathcal{C}\cap\{A\in\mathbb{F}_q^{m\times n}\mid \wt_R(A-E)\leq t\}\rvert\geq N,$$
see, e.g., \cite{ding2014list}. Hence, asymptotically, we will only have one solution in the list.
One would then go through the list and check whether $\text{wt}_C(x_i)\leq t.$ For one such $x_i$ the cost of computing the cover weight is in $O\left((m+n)n^2\right)$.
 However, the generic decoders in the rank metric have a cost that is much larger than our cost. For example \cite{grs} has an asymptotic cost of $q^{n^2TR},$ for 
\begin{align*}
    R := \lim_{n \to \infty} \frac{k(n)}{n} \in (0,1), \
    T := \lim_{n \to \infty} \frac{t(n)}{f(n)} \in (0,1).
\end{align*}
In  \cite{bardet}, the authors provided an algorithm to solve the MinRank problem (Problem \ref{minrank}). The cost of this algorithm is however at least  $\binom{n}{t}^2$, which is much larger than the cost of our algorithm, which is less than $\binom{m+n}{t}.$ Thus, we conclude that using a MinRank solver on a cover-metric decoding instance will not provide a faster algorithm than our newly proposed generic decoder in the cover metric. 

\section{Conclusion}\label{sec:conclusion}
The cover metric lies in between the Hamming and the rank weight and could thus potentially profit from the benefits of both metrics. 
In this paper, we provided a polynomial-time reduction from the decoding problem in the Hamming metric to the decoding problem in the cover metric and, hence, showed the NP-hardness of the decoding problem for $\mathbb{F}_q$-linear codes. 
This similarity to the Hamming metric and the rank metric could enable cryptographic applications. 

Also note that random codes in the cover metric have with high probability the largest possible minimum cover distance, which makes this metric unique. 

To assess the cost of the cover-metric decoding problem, we provided a generic decoder and analyzed its cost. An asymptotic analysis shows that the decoding problem in the cover metric is exponential in the number of lines. This behaviour is completely different to the Hamming metric and the rank metric, which have a complexity that is exponential in the number of code symbols. 

\subsection*{Acknowledgements}
The first author acknowledges the financial support by the Federal Ministry of Education and Research of Germany in the programme of "Souverän. Digital. Vernetzt.". Joint project 6G-life, project identification number: 16KISK002.
The fourth author  is  supported by the Swiss National Science Foundation grant number 195290.

\bibliographystyle{plain}
\bibliography{references}

\begin{thebibliography}{10}

\bibitem{bardet}
Magali Bardet, Maxime Bros, Daniel Cabarcas, Philippe Gaborit, Ray Perlner,
  Daniel Smith-Tone, Jean-Pierre Tillich, and Javier Verbel.
\newblock Improvements of algebraic attacks for solving the rank decoding and
  {MinRank} problems.
\newblock In {\em International Conference on the Theory and Application of
  Cryptology and Information Security}, pages 507--536. Springer, 2020.

\bibitem{barg}
Alexander Barg.
\newblock Some new {NP}-complete coding problems.
\newblock {\em Problemy Peredachi Informatsii}, 30(3):23--28, 1994.

\bibitem{berlekamp1978inherent}
Elwyn Berlekamp, Robert McEliece, and Henk Van~Tilborg.
\newblock On the inherent intractability of certain coding problems (corresp.).
\newblock {\em IEEE Transactions on Information Theory}, 24(3):384--386, 1978.

\bibitem{minrank}
Jonathan~F Buss, Gudmund~S Frandsen, and Jeffrey~O Shallit.
\newblock The computational complexity of some problems of linear algebra.
\newblock {\em Journal of Computer and System Sciences}, 58(3):572--596, 1999.

\bibitem{thomas}
Andr{\'e} Chailloux, Thomas Debris-Alazard, and Simona Etinski.
\newblock Classical and quantum algorithms for generic syndrome decoding
  problems and applications to the {L}ee metric.
\newblock In {\em International Conference on Post-Quantum Cryptography}, pages
  44--62. Springer, 2021.

\bibitem{delsarte}
Philippe Delsarte.
\newblock Bilinear forms over a finite field, with applications to coding
  theory.
\newblock {\em Journal of Combinatorial Theory, Series A}, 25(3):226--241,
  1978.

\bibitem{ding2014list}
Yang Ding.
\newblock On list-decodability of random rank metric codes and subspace codes.
\newblock {\em IEEE Transactions on Information Theory}, 61(1):51--59, 2014.

\bibitem{gabidulin1985optimal}
Ernest~M. Gabidulin.
\newblock Optimal array error-correcting codes.
\newblock {\em Probl. Peredach. Inform.}, 21(2):102--106, 1985.

\bibitem{gabidulin}
Ernest~M. Gabidulin.
\newblock Theory of codes with maximum rank distance.
\newblock {\em Problemy Peredachi Informatsii}, 21(1):3--16, 1985.

\bibitem{gabidulin1972lattice}
Ernest~M. Gabidulin and Valery~I. Korzhik.
\newblock Codes correcting lattice-pattern errors.
\newblock {\em Zzvestiya VUZ. Radioelektronika}, 4(6):7, 1972.

\bibitem{grs}
Philippe Gaborit, Olivier Ruatta, and Julien Schrek.
\newblock On the complexity of the rank syndrome decoding problem.
\newblock {\em IEEE Transactions on Information Theory}, 62(2):1006--1019,
  2015.

\bibitem{ranknp}
Philippe Gaborit and Gilles Z{\'e}mor.
\newblock On the hardness of the decoding and the minimum distance problems for
  rank codes.
\newblock {\em IEEE Transactions on Information Theory}, 62(12):7245--7252,
  2016.

\bibitem{horlemann2021information}
Anna-Lena Horlemann, Sven Puchinger, Julian Renner, Thomas Schamberger, and
  Antonia Wachter-Zeh.
\newblock Information-set decoding with hints.
\newblock In {\em International Workshop on Code-Based Cryptography
  (CBCrytpo)}, 2021.

\bibitem{Z4}
Anna-Lena Horlemann-Trautmann and Violetta Weger.
\newblock Information set decoding in the {L}ee metric with applications to
  cryptography.
\newblock {\em Advances in Mathematics of Communications}, 15(4):677--699,
  2021.

\bibitem{liu2018list}
Shu Liu, Chaoping Xing, and Chen Yuan.
\newblock List decoding of cover metric codes up to the {S}ingleton bound.
\newblock {\em IEEE Transactions on Information Theory}, 64(4):2410--2416,
  2018.

\bibitem{martinez2022multilayer}
Umberto Mart{\'\i}nez-Pe{\~n}as.
\newblock Multilayer crisscross error and erasure correction.
\newblock {\em arXiv preprint arXiv:2203.07238}, 2022.

\bibitem{ROLLO}
Carlos~Aguilar Melchor, Nicolas Aragon, Magali Bardet, Slim Bettaieb, Lo{\"\i}c
  Bidoux, Olivier Blazy, and Jean-Christophe Deneuville.
\newblock {ROLLO-R}ank-{O}uroboros, {LAKE} \& {LOCKER}:submission to the {NIST}
  post-quantum cryptography call, 2019.

\bibitem{RQC}
Carlos~Aguilar Melchor, Nicolas Aragon, Slim Bettaieb, Lo{\"\i}c Bidoux,
  Olivier Blazy, Jean-Christophe Deneuville, Philippe Gaborit, and Gilles
  Z{\'e}mor.
\newblock Rank quasi cyclic ({RQC}): submission to the {NIST} post-quantum
  cryptography call, 2017.

\bibitem{roth2}
Ron~M Roth.
\newblock Maximum-rank array codes and their application to crisscross error
  correction.
\newblock {\em IEEE transactions on Information Theory}, 37(2):328--336, 1991.

\bibitem{roth}
Ron~M Roth.
\newblock Probabilistic crisscross error correction.
\newblock {\em IEEE Transactions on Information Theory}, 43(5):1425--1438,
  1997.

\bibitem{Wachterzeh_Crisscross_journal}
Antonia {Wachter-Zeh}.
\newblock {List Decoding of Crisscross Errors}.
\newblock {\em IEEE Transactions on Information Theory}, 63(1):142--149,
  January 2017.

\bibitem{leenp}
Violetta Weger, Karan Khathuria, Anna-Lena Horlemann-Trautmann, Massimo
  Battaglioni, Paolo Santini, and Edoardo Persichetti.
\newblock On the hardness of the {L}ee syndrome decoding problem.
\newblock {\em Advances in Mathematics of Communications}, 2022.

\end{thebibliography}

\end{document}